\documentclass[]{IEEEtran}
\usepackage{cite}
\usepackage{amsmath,amssymb,amsfonts}
\usepackage{algorithmic}
\usepackage{graphicx}
\usepackage{textcomp}

\newtheorem{theorem}{Theorem}
\newtheorem{lem}{Lemma}
\newtheorem{remark}{Remark}

\newcommand{\proof}{\noindent \textit{Proof}: }

\def\BibTeX{{\rm B\kern-.05em{\sc i\kern-.025em b}\kern-.08em
		T\kern-.1667em\lower.7ex\hbox{E}\kern-.125emX}}

\begin{document}

	\title{
Robust Energy Shaping Control of an Underactuated Inverted Pendulum}
\author{ M. Reza J. Harandi and Mehrzad Namvar
}

\maketitle
\begin{abstract}
Although the stabilization of underactuated systems remains a challenging problem, the total energy shaping approach provides a general framework for addressing this objective. However, the practical implementation of this method is hindered by the need to analytically solve a set of partial differential equations (PDEs), which constitutes a major obstacle. In this paper, a rotary inverted pendulum system is considered, and an interconnection and damping assignment passivity-based control (IDA-PBC) scheme is developed by deriving concise analytical solutions to the kinetic and potential energy PDEs. Furthermore, a novel robust term is incorporated into the control law to compensate for a specific class of disturbances that has not been addressed within the existing IDA-PBC literature. The effectiveness of the proposed method is validated through numerical simulations, demonstrating satisfactory control performance. 
\end{abstract}

\maketitle
\section{INTRODUCTION}\label{s1}
Stabilization of underactuated robots (UaRs) has attracted considerable attention due to the inherent challenge arising from the mismatch between the number of actuators and the system degrees of freedom. Over the past two decades, numerous control strategies have been proposed to address this problem. However, a large portion of the existing methods are inherently case-specific and rely heavily on the particular structure or dynamics of the system under consideration, which limits their general applicability.
Among the available control frameworks, interconnection and damping assignment passivity-based control (IDA-PBC) has emerged as a systematic and general approach for stabilizing underactuated robotic systems through total energy shaping. By appropriately modifying the system’s kinetic and potential energy functions, IDA-PBC enables the stabilization of unstable equilibria while preserving the underlying physical structure of the system. Although various formulations of IDA-PBC have been successfully applied to several benchmark UaRs reported in the literature, relatively little attention has been devoted to the rotary inverted pendulum (RIP), and the validity and applicability of the existing results in this context remain subject to further investigation.

Potential energy shaping has proven to be a successful approach for stabilizing fully actuated mechanical systems. However, due to its inherent inability to stabilize UaRs, the concept of total energy shaping was introduced, in which both the kinetic and potential energies of the system are modified in the closed-loop dynamics through controller design.
Within this framework, IDA-PBC provides a systematic methodology for assigning the desired Hamiltonian function, as well as the interconnection and damping matrices of the system. The realization of this objective relies on analytically solving a set of partial differential equations (PDEs) associated with kinetic and potential energy shaping. These equations, commonly referred to as matching equations, constitute the main bottleneck of the IDA-PBC approach and significantly restrict its practical applicability.
To overcome this limitation, several studies have been reported in the literature with the aim of simplifying the matching equations or eliminating the need for their explicit solution for specific classes of UaRs.

In \cite{acosta2005interconnection}, the authors focused on UaRs with underactuated degrees one. For a specific class of UaRs satisfying a set of structural assumptions, the kinetic and potential energy shaping PDEs were reduced to ordinary differential equations (ODEs). The simplification of the kinetic energy PDE through a coordinate transformation was investigated in \cite{viola2007total}, where it was shown that, under a certain condition, the forcing term can be eliminated. In \cite{donaire2015shaping}, it was demonstrated that total energy shaping can be achieved without explicitly solving the matching equations, provided that specific conditions are satisfied. The work in \cite{harandi2021matching} concentrated on the kinetic energy PDE and showed that, by imposing a particular structure on the desired inertia matrix, the PDE can be significantly simplified or reduced to an ODE. Other representative studies addressing the matching equations within the total energy shaping framework include \cite{acosta2009pdes, harandi2022solution, borja2015shaping}; see \cite{harandi2021passivity} for a comprehensive review.
Despite these contributions, the applicability of the aforementioned approaches to RIP remains limited, primarily due to the restrictive assumptions and structural requirements imposed by these methods. Hence, to the best of our knowledge, none of the existing general approaches is capable of shaping the energy of RIP; see Remark~\ref{re1} for more details.
 In \cite{ryalat2013ida}, the authors focused on RIP and claimed that the kinetic energy shaping PDE associated with RIP reduces to an algebraic equation and subsequently designed an IDA-PBC scheme. However, a closer examination reveals that the resulting equation is, in fact, an ODE, and the proposed desired inertia matrix does not satisfy this equation. Consequently, the validity of the reported results is not fully established; see Remark~\ref{re2} for further discussion.

On the other hand, robust extensions of IDA-PBC have attracted increasing attention in recent years. In \cite{donaire2017robust}, a constant matched disturbance was rejected by introducing an integral action into the closed-loop system through. The rejection of external disturbances modeled as the output of a port-Hamiltonian system via the design of an estimator was investigated in \cite{ferguson2020matched}. Moreover, the rejection of a class of matched disturbances that are position-dependent with unknown constant parameters has been addressed in \cite{franco2025integral}. In these works, closed-loop stability is established under the assumption that the disturbance is integrable, that is, it can be expressed as the gradient of a scalar function, which inherently restricts the range of admissible disturbances and consequently, limits the applicability of the proposed approaches.

In this paper, the stabilization of the upward equilibrium of a rotary inverted pendulum is addressed within the total energy shaping framework. By invoking the approach presented in~\cite{sandoval2008interconnection}, the matching equation associated with kinetic energy shaping is structurally simplified, and an appropriate analytical solution is derived that ensures the positive definiteness of the Hessian matrix of the desired potential energy. Subsequently, by solving the remaining PDE related to potential energy shaping, it is shown that the region of attraction of the closed-loop system can be enlarged to approach the largest feasible set through a suitable selection of the controller parameters. In addition, the existing works related to energy shaping control of the rotary inverted pendulum are thoroughly analyzed and discussed.
Furthermore, a robust term is incorporated into the control law to reject a class of matched disturbances. Specifically, a nonintegrable position-dependent disturbance with unknown constant parameters, expressed in a regressor form, is eliminated from the closed-loop dynamics, and stability is guaranteed using Lyapunov theory. The effectiveness of the proposed control strategy is validated through numerical simulations.   

For clarity, the main contributions of this paper are summarized as follows:

\begin{itemize}
	\item A rigorous analytical solution to the matching equations is derived. 
	The PDEs associated with both kinetic and potential energy shaping are solved analytically, ensuring exact matching of the unactuated dynamics in the open- and closed-loop systems, which is essential for the performance of IDA-PBC.
	
	\item A comprehensive analysis of the existing literature on energy shaping control of the RIP is provided. 
	It is shown that, in prior works, the matching equations are not consistently satisfied, which may lead to violations of closed-loop stability. This analysis highlights the necessity of a precise total energy shaping design for the RIP.
	
	\item A novel disturbance rejection mechanism is developed within the IDA-PBC framework. 
	The proposed approach compensates for a class of nonintegrable matched disturbances with unknown constant parameters, expressed in regressor form, and guarantees closed-loop stability using Lyapunov theory.
\end{itemize}
\section{IDA-PBC framework for UaRs}\label{sec:21}

This part briefly reviews the IDA-PBC framework for UaRs, following the developments reported in~\cite{harandi2021matching,acosta2005interconnection}. 
The dynamics of a simple mechanical system described within the port-Hamiltonian formalism can be expressed as
\begin{equation}
	\label{1}
	\begin{bmatrix}
		\dot{q} \\ \dot{p}
	\end{bmatrix}
	=
	\begin{bmatrix}
		0_{n\times n} & I_n \\
		- I_n & 0_{n\times n}
	\end{bmatrix}
	\begin{bmatrix}
		\nabla_q H \\ \nabla_p H
	\end{bmatrix}
	+
	\begin{bmatrix}
		0_{n\times m} \\ G
	\end{bmatrix}
	u,
\end{equation}
where $q,p \in \mathbb{R}^n$ denote the generalized coordinates and momenta, respectively, with $p = M(q)\dot{q}$. 
The Hamiltonian function is given by
\[
H(q,p) = V(q) + \frac{1}{2} p^\top M^{-1}(q) p,
\]
which represents the total energy of the system composed of potential energy $V(q)$ and kinetic energy. 
The inertia matrix $M(q)$ is symmetric and positive definite, and the input distribution matrix is denoted by $G(q)\in\mathbb{R}^{n\times m}$, where $m<n$ indicates underactuation.

Within the IDA-PBC methodology, the control input is selected as
\begin{align}
	u &= (G^\top G)^{-1}G^\top\big(\nabla_q H - M_d M^{-1}\nabla_q H_d + J_2\nabla_p H_d\big) \nonumber\\
	&\quad - K_v G^\top \nabla_p H_d,
	\label{4}
\end{align}
where the desired Hamiltonian is defined by
\[
H_d(q,p) = \frac{1}{2} p^\top M_d^{-1}(q)p + V_d(q).
\]
The desired inertia matrix $M_d(q)$ and the shaped potential energy $V_d(q)$ are obtained by enforcing the so-called matching equations associated with kinetic and potential energy shaping, namely,
\begin{align}
	&G^\bot \Big\{\nabla_q \big(p^\top M^{-1}(q)p\big)
	- M_d M^{-1}\nabla_q \big(p^\top M_d^{-1}(q)p\big) \nonumber\\
	&\qquad\quad + 2 J_2 M_d^{-1} p \Big\} = 0_{(n-m)\times 1},
	\label{5}\\
	&G^\bot \{\nabla_q V - M_d M^{-1}\nabla_q V_d\}
	= 0_{(n-m)\times 1}.
	\label{6}
\end{align}

As a result, the closed-loop dynamics can be written in the compact form
\begin{equation}
	\label{2}
	\begin{bmatrix}
		\dot{q} \\ \dot{p}
	\end{bmatrix}
	=
	\begin{bmatrix}
		0_{n\times n} & M^{-1}M_d \\
		- M_d M^{-1} & J_2 - G K_v G^\top
	\end{bmatrix}
	\begin{bmatrix}
		\nabla_q H_d \\ \nabla_p H_d
	\end{bmatrix}.
\end{equation}
Here, $J_2(q,p)\in\mathbb{R}^{n\times n}$ is a freely assignable skew-symmetric matrix, and $K_v\in\mathbb{R}^{m\times m}$ is a positive definite damping gain. 
The derivation of~\eqref{2} follows directly by premultiplying the momentum dynamics in~\eqref{1} with $[G, G^{\bot^\top}]^\top$, substituting the control law~\eqref{4}, and invoking the matching equations~\eqref{5}–\eqref{6}; see~\cite{ortega2002stabilization} for a detailed discussion.

The desired equilibrium point $q^\ast$ lies on the manifold defined by $G^\bot \nabla_q V(q^\ast)=0$ and is required to be a strict minimum of the shaped potential energy $V_d(q)$. 
Under this condition, stability of the equilibrium is ensured since the time derivative of the desired Hamiltonian satisfies
\begin{align}
	\dot{H}_d = -(\nabla_p H_d)^\top G K_v G^\top \nabla_p H_d \leq 0.
	\label{vdot}
\end{align}
\begin{figure}[b]
	\centering
	\includegraphics[scale=.3]{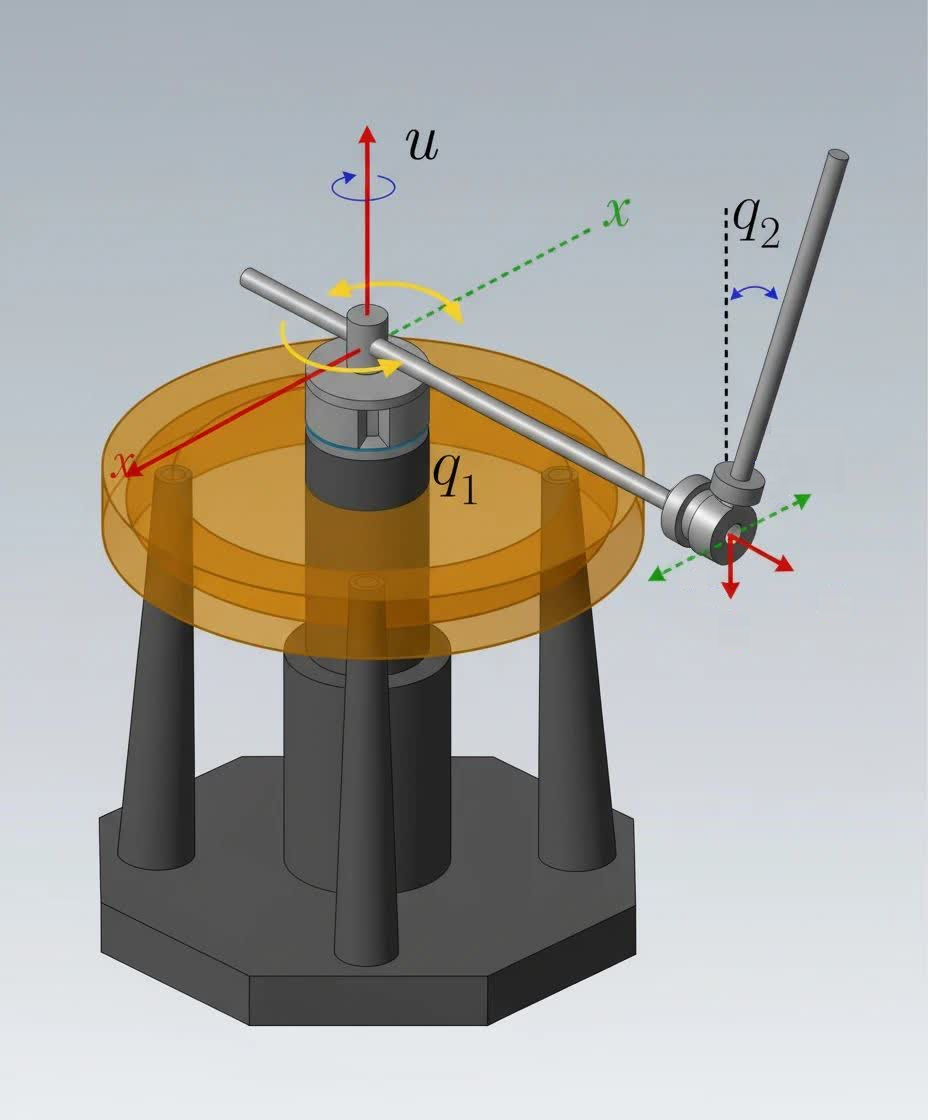}
	\caption
	{Schematic of the rotary inverted pendulum.}
	\label{sh1}
\end{figure}

\section{Total Energy Shaping Controller Design for the RIP}

The rotary inverted pendulum  is a spatial robotic system consisting of two revolute joints, in which only the first joint is actuated. 
A schematic representation of the RIP is depicted in Fig.~\ref{sh1}. 
The dynamics of the RIP can be expressed in the form of~\eqref{1} with the following system parameters~\cite{hernandez2024modeling}:
\begin{align}\label{dyn}
	M &=
	\begin{bmatrix}
		p_1 + p_2 \sin^2(q_2) & p_3 \cos(q_2) \\
		p_3 \cos(q_2) & p_4
	\end{bmatrix}, \qquad
	G =
	\begin{bmatrix}
		1 \\ 0
	\end{bmatrix}, \\
	V &= p_5 \cos(q_2),
\end{align}
where the parameters are defined as
\begin{align*}
	p_1 &= I_1 + m_1 l_1^2, \qquad
	p_2 = m_2 l_2^2, \qquad
	p_3 = m_2 l_1 l_2, \\
	p_4 &= I_2 + m_2 l_2^2, \qquad
	p_5 = m_2 l_2 g,
\end{align*}
in which $I_i$, $l_i$, and $m_i$ denote the moment of inertia, length, and mass of the $i$th link, respectively. The aim is to stabilize $q^\ast=[0,0]^\top$ via IDA-PBC. The results are presented in the following theorem.

\begin{theorem}\label{th1}
	Consider the rotary inverted pendulum whose dynamics are described by \eqref{1} and \eqref{dyn}. To stabilize the equilibrium point $q^\ast$, a total energy shaping controller of the form \eqref{4} is applied. The controller parameters are chosen as follows
	\begin{align}
		M_d(q_2) &=
		\begin{bmatrix}
			d_1 & d_2 \\
			d_3 & d_4
		\end{bmatrix}, \qquad
		J_2 =
		\begin{bmatrix}
			0 & \tilde{p}^\top \alpha(q_2) \\
			-\tilde{p}^\top \alpha(q_2) & 0
		\end{bmatrix}, \nonumber\\
		\tilde{p} &= M_d^{-1} p, \qquad
		\alpha(q_2) = [\alpha_1(q_2), \alpha_2(q_2)]^\top,
		\label{j2}
	\end{align}
	where
	\begin{align}
		d_1 &= k_2, \nonumber\\
		d_2 = d_3 &=
		\cos(q_2)
		\frac{p_1 + p_2 \sin^2(q_2)}
		{k_1 + \dfrac{p_2}{p_3 \psi_{4_0}} \sin^2(q_2)}
		- p_3 \psi_{4_0} \cos(q_2), \nonumber\\
		d_4 &=
		\frac{p_3 \cos^2(q_2)}
		{k_1 + \dfrac{p_2}{p_3 \psi_{4_0}} \sin^2(q_2)}
		- p_4 \psi_{4_0}.
		\label{d}
	\end{align}
	The desired potential energy function is given by
	\begin{align}
		V_d &=
		\frac{\kappa}{2}
		\bigg(
		q_1
		+
		\sqrt{\frac{p_3}{k_1 p_2 \psi_{4_0}}}
		\arctan
		\Big(
		\sqrt{\frac{p_2}{k_1 p_3 \psi_{4_0}}}
		\, \sin(q_2)
		\Big)
		\bigg)^2
		\nonumber\\
		&\quad
		- \frac{p_5}{\psi_{4_0}} \cos(q_2).
		\label{vd}
	\end{align}
	Moreover, the components of $\alpha(q_2)$ are defined as
	\begin{align}
		2\alpha_1& =
		- 2 p_2 \psi_1^2 \sin(q_2) \cos(q_2)
		+ 2 p_3 \psi_1 \psi_2 \sin(q_2)
		\nonumber\\
		&
		+ \psi_4 (p_1 + p_2 \sin^2(q_2)) \frac{d\psi_1}{dq_2}
		- p_3 \psi_4 \psi_2 \sin(q_2)
		\nonumber\\
		&
		+ 2 p_2 \psi_4 \psi_1 \sin(q_2) \cos(q_2)
		+ p_3 \psi_4 \cos(q_2) \frac{d\psi_2}{dq_2},
		\label{alpha}
		\\
		\alpha_2 &=
		p_3 \psi_2 \psi_3 \sin(q_2)
		- 2 p_2 \psi_1 \psi_3 \sin(q_2) \cos(q_2)
		\nonumber\\
		& 	+ p_3 \psi_1 \psi_4 \sin(q_2)
		+ p_3 \psi_4 \cos(q_2) \frac{d\psi_1}{dq_2}
		+ p_4 \psi_4 \frac{d\psi_2}{dq_2}\nonumber\\
		& - p_3 \psi_4 \psi_1 \sin(q_2).
		\nonumber
	\end{align}
	Here, $\psi_{4_0}$, $k_1$, $k_2$, and $\kappa$ are positive constants, and
	\begin{align}\label{psi}
		M_d M^{-1} = \Psi(q_2) =
		\begin{bmatrix}
			\psi_1(q_2) & \psi_2(q_2) \\
			\psi_3(q_2) & \psi_4(q_2)
		\end{bmatrix}.
	\end{align}
\end{theorem}

\proof
	To shape the energy of the RIP, we first solve the PDE given in \eqref{5}. 
	To simplify this equation, a procedure similar to that in \cite{sandoval2008interconnection} is adopted. 
	Since the inertia matrix depends only on $q_2$, it is convenient to select $M_d$ and $J_2$ to be independent of $q_1$. 
	Using the identity $M M^{-1} = I_2$, differentiation with respect to $q_2$ yields
	\begin{equation*}
		\frac{d M^{-1}}{d q_2}
		=
		- M^{-1} \frac{d M}{d q_2} M^{-1}.
	\end{equation*}
	Moreover, considering $J_2$ as defined in \eqref{j2}, equation \eqref{5} can be written as
	\begin{align*}
		&- p^\top M^{-1} \frac{d M}{d q_2} M^{-1} p
		+ G^\top M_d M^{-1}
		\begin{bmatrix}
			0 \\ p^\top M_d^{-1} \dfrac{d M_d}{d q_2} M_d^{-1} p
		\end{bmatrix} \\
		&\qquad
		- 2
		\begin{bmatrix}
			\tilde{p}^\top \alpha(q_2) & 0
		\end{bmatrix}
		M_d^{-1} p
		= 0 .
	\end{align*}
	By rearranging the last two terms, the above equation can be equivalently expressed as
	\begin{align*}
		p^\top \Bigg(
		&- M^{-1} \frac{d M}{d q_2} M^{-1}
		+ [0,1] M_d M^{-1} [0,1]^\top
		M_d^{-1} \frac{d M_d}{d q_2} M_d^{-1} \\
		&\quad
		- 2 M_d^{-1}
		\begin{bmatrix}
			\alpha_1 & 0 \\
			\alpha_2 & 0
		\end{bmatrix}
		M_d^{-1}
		\Bigg) p
		= 0 .
	\end{align*}
	Eliminating $p$, taking the symmetric part of the last term, and premultiplying and postmultiplying by $M_d$, the resulting equation is obtained as
	\begin{align}\label{kin}
		- \Psi \frac{d M}{d q_2} \Psi^\top
		+ \big( [0,1] \Psi [0,1]^\top \big)
		\frac{d M_d}{d q_2}
		-
		\begin{bmatrix}
			2\alpha_1 & \alpha_2 \\
			\alpha_2 & 0
		\end{bmatrix}
		= 0,
	\end{align}
	where \eqref{psi} has been used.
   Replacing \eqref{dyn} into \eqref{kin} yields the following set of equations:
\begin{subequations} \begin{align} &2p_3\psi_1\psi_2\sin(q_2)-2p_2\psi_1^2\sin(q_2)\cos(q_2)+\psi_4\frac{d}{dq_2}\Big(\psi_1\big(p_1\nonumber\\&\hspace{1.cm}+p_2\sin^2(q_2)\big)+p_3\psi_2\cos(q_2)\Big)-2\alpha_1=0,\label{al1}\\& p_3\sin(q_2)(\psi_2\psi_3+\psi_1\psi_4)-2p_2\psi_1\psi_3\sin(q_2)\cos(q_2)\nonumber\\&\hspace{1.1cm}+\psi_4\frac{d}{dq_2}\big(p_3\psi_1\cos(q_2)+p_4\psi_2\big)-\alpha_2=0,\label{al2}\\& -2 p_2 \psi_3^2\sin(q_2) \cos(q_2)+2 p_3\psi_3\psi_4 \sin(q_2)\nonumber\\&\hspace{1.8cm}+\psi_4\frac{d}{dq_2}\big(p_3\psi_3\cos(q_2)+p_4\psi_4\big)=0.\label{ode} \end{align} \end{subequations}
   Since the first two equations contain the free parameters $\alpha_1$ and $\alpha_2$, it suffices to solve only the last equation. 
   Note that \eqref{ode} is an ordinary differential equation with two unknown functions $\psi_3$ and $\psi_4$, and therefore admits infinitely many solutions. 
   To uniquely determine $\psi_i$s, the following conditions are imposed:
   
   \begin{enumerate}
   	\item[(C1)] $M_d$ must be positive definite at least in a neighborhood of $q^\ast$.
   	\item[(C2)] The solution of \eqref{6}, given by
   	\begin{align}\label{pot}
   		- p_5 \sin(q_2)
   		=
   		\psi_3(q_2) \frac{\partial V_d}{\partial q_1}
   		+
   		\psi_4(q_2) \frac{\partial V_d}{\partial q_2},
   	\end{align}
   	must satisfy
   	\begin{align}\label{cn}
   		\left.
   		\frac{\partial^2 V_d}{\partial q^2}
   		\right|_{q=q^\ast}
   		> 0 .
   	\end{align}
   \end{enumerate}
   
   From \eqref{psi}, the elements of $M_d$ are obtained as
   \begin{subequations}\label{eq:main}
   	\begin{align}
   		d_1 &=
   		\psi_1 (p_1 + p_2 \sin^2(q_2))
   		+ p_3 \psi_2 \cos(q_2),
   		\label{eq:a}
   		\\
   		d_2 &=
   		p_3 \psi_1 \cos(q_2)
   		+ p_4 \psi_2,
   		\label{eq:b}
   		\\
   		d_3 &=
   		\psi_3 (p_1 + p_2 \sin^2(q_2))
   		+ p_3 \psi_4 \cos(q_2),
   		\label{eq:c}
   		\\
   		d_4 &=
   		p_3 \psi_3 \cos(q_2)
   		+ p_4 \psi_4.
   		\label{eq:d}
   	\end{align}
   \end{subequations}
   A necessary condition for (C1) is $d_4 > 0$. 
   On the other hand, to satisfy \eqref{cn}, the Hessian matrix of the nonhomogeneous solution $V_{dn}$ which is as follows
   \begin{align*}
   	V_{dn}
   	=
   	\int
   	\frac{- p_5 \sin(q_2)}{\psi_4(q_2)}
   	\, dq_2,
   \end{align*}
   must be positive definite at $q = q^\ast$ \cite{harandi2022solution}. 
   Accordingly, choosing
   \begin{align}\label{psi4}
   	\psi_4 = - \psi_{4_0},
   \end{align}
   with $\psi_{4_0} > 0$ guarantees
   $\nabla_q^2 V_{dn} |_{q = q^\ast} > 0$.
   Substituting \eqref{psi4} into \eqref{ode} yields
   \begin{align}
   	\psi_3'
   	=
   	- \tan(q_2)\, \psi_3
   	- \frac{2 p_2}{p_3 \psi_{4_0}} \sin(q_2) \psi_3^2,
   	\label{r}
   \end{align}
   where $\psi_3' = d\psi_3/dq_2$. 
   To solve the Riccati equation \eqref{r}, consider the substitution
   \begin{align*}
   	\psi_3(q_2) = \frac{1}{\eta(q_2)},
   \end{align*}
   which leads to
   \begin{align*}
   	\eta'
   	=
   	\eta \tan(q_2)
   	+ \frac{2 p_2}{p_3 \psi_{4_0}} \sin(q_2).
   \end{align*} 
   Multiplying both sides by the integrating factor
   \begin{align*}
   	\mu(q_2)
   	=
   	e^{\int -\tan(q_2) \, dq_2}
   	=
   	\cos(q_2),
   \end{align*}
   we obtain
   \begin{align*}
   	\big( \cos(q_2)\, \eta \big)'
   	=
   	\frac{2 p_2}{p_3 \psi_{4_0}}
   	\sin(q_2)\cos(q_2).
   \end{align*}
   Hence,
   \begin{align*}
   	\eta(q_2)
   	=
   	\frac{k_1 + \dfrac{p_2}{p_3 \psi_{4_0}} \sin^2(q_2)}
   	{\cos(q_2)},
   \end{align*}
   where $k_1 > 0$ is an integration constant. 
   Substituting back yields
   \begin{align}\label{psi3}
   	\psi_3(q_2)
   	=
   	\frac{\cos(q_2)}
   	{k_1 + \dfrac{p_2}{p_3 \psi_{4_0}} \sin^2(q_2)}.
   \end{align}
   Substituting \eqref{psi3} and \eqref{psi4} into \eqref{eq:c} and \eqref{eq:d} gives $d_3$ and $d_4$ as in \eqref{d}. 
   Next, the free parameters $\psi_1$ and $\psi_2$ are designed. 
   Imposing the conditions $d_2 = d_3$ and $d_1 > 0$, we select
   \begin{align*}
   	d_1 = k_2, \qquad k_2 > 0 .
   \end{align*}
   
   From \eqref{eq:a} and \eqref{eq:b}, it follows that
   \begin{align*}
   	\psi_1(q_2)
   	&=
   	\frac{
   		p_4 d_1 - p_3 \cos(q_2)\, d_2
   	}{\det(M)},
   	\\
   	\psi_2(q_2)
   	&=
   	\frac{
   		- p_3 \cos(q_2)\, d_1
   		+ (p_1 + p_2 \sin^2(q_2)) d_2
   	}{\det(M)}.
   \end{align*}
   Finally, $\alpha_1$ and $\alpha_2$ are obtained from \eqref{al1} and \eqref{al2}, respectively, as given in \eqref{alpha}. 
   
   The PDE \eqref{6} is a first-order linear PDE whose solution can be obtained by standard methods as follows
   \begin{align*}
   	V_d
   	&=
   	\Phi
   	\bigg(
   	q_1
   	+
   	\sqrt{\frac{p_3}{k_1 p_2 \psi_{4_0}}}
   	\arctan
   	\Big(
   	\sqrt{\frac{p_2}{k_1 p_3 \psi_{4_0}}}
   	\, \sin(q_2)
   	\Big)
   	\bigg)\\&
   	-
   	\frac{p_5}{\psi_{4_0}} \cos(q_2),
   \end{align*}
   where $\Phi(\cdot)$ is an arbitrary function. 
   To satisfy \eqref{cn}, $V_d$ is chosen as in \eqref{vd} with $\kappa > 0$, which completes the proof.
   
   \hfill$\square$

Next, we determine the region in which the desired inertia matrix $M_d$ is positive definite.
After straightforward calculations, it can be shown that $d_4 > 0$ holds provided that
\begin{align}
	q_2 \in [-\rho, \, \rho], \qquad
	\rho
	=
	\arccos
	\sqrt{
		\frac{
			p_4 \big( \psi_{4_0} k_1 + \tfrac{p_2}{p_3} \big)
		}{
			p_3 + \tfrac{p_2 p_4}{p_3}
		}
	}.\label{rho}
\end{align}
This result implies that $d_4$ remains positive in a neighborhood of the equilibrium point $q^\ast$.
Moreover, the admissible region enlarges as the product $\psi_{4_0} k_1$ decreases.
Finally, by choosing $k_2$ sufficiently large, the interval
$q_2 \in [-\rho, \, \rho]$ characterizes the region in which $M_d$ is positive definite.

\begin{remark}\label{re1}
	As stated in Section~\ref{s1}, several studies have been devoted to solving or simplifying the matching equations arising in energy shaping control. 
	However, these approaches are not applicable to the RIP. 
	The method proposed in \cite{acosta2005interconnection} is based on canceling the forcing term, i.e., the first term in \eqref{5}, which does not vanish for the RIP. 
	Moreover, the necessary condition introduced in \cite{viola2007total} for eliminating the forcing term via a coordinate transformation is not satisfied in this case. 
	The systematic approach developed in \cite{harandi2021matching} for solving the kinetic energy PDE results in a Riccati equation that does not admit a closed-form solution for the RIP. 
	In addition, the sufficient conditions proposed in \cite{donaire2015shaping} for energy shaping without explicitly solving the PDEs cannot be fulfilled for this system. 
	Furthermore, due to the specific structure of the inertia matrix, stabilization via PID-PBC is not feasible \cite{romero2018global}. 
Therefore, to the best of our knowledge, the existing general approaches fail to address the energy shaping problem of the RIP, rendering it a particularly challenging benchmark system.
\end{remark}

   \begin{remark}\label{re2}
   	Although general energy shaping methods are not applicable to the RIP, some attempts have been reported to stabilize this system using the IDA-PBC framework. 
   	In \cite{ryalat2013ida}, the desired inertia matrix is selected arbitrarily, and it is claimed that \eqref{5} reduces to an algebraic equation that can be readily solved by an appropriate choice of $J_2$. 
   	However, it has been shown that the forcing term does not vanish for the RIP and, as a consequence, the desired inertia matrix must be obtained by solving an ODE. 
   	This observation casts doubt on the validity of the reported results.
   	Another energy shaping controller for the RIP is presented in \cite{viola2013some}. 
   	In that work, a coordinate transformation is first applied and the resulting ODE is subsequently solved. 
   	In particular, it is claimed that the solution of
   	\begin{equation*}
   		\mathfrak{k}_1 \frac{d m_{22}}{d q_2}
   		=
   		- \sin(2 q_2)\, m_{22}^2
   		- 4\, m_{22}
   		+ \frac{2 \mathfrak{k}_1}{\cos^2 q_2}
   	\end{equation*}
   	is given by
   	\begin{equation*}
   		m_{22}(q_2)
   		=
   		\frac{2 \mathfrak{k}_1}{b^2 \cos^2 q_2}
   		+
   		\frac{\mathfrak{k}_1}{\mathfrak{k}_2^2 + \sin^2 q_2}.
   	\end{equation*}
   	However, direct substitution reveals that the proposed expression for $m_{22}$ does not satisfy the above differential equation. 
   	Therefore, to the best of our knowledge, this paper presents the first precise IDA-PBC design for the RIP.
   \end{remark}
   At the end of this section, it is worth emphasizing that the specific selection of $\psi_4$ plays a key role in overcoming this difficulty. 
   In particular, choosing $\psi_4$ in the proposed form enables the resulting Riccati-type equation to admit a tractable solution, which is not achievable under a general choice of the desired inertia parameters.
   
    \section{Disturbance Rejection}
    In this section, the controller is redesigned to reject a particular class of matched disturbances acting on the RIP. 
    Consider the disturbed dynamics of the system given by
    \begin{equation}
    	\label{dis}
    	\begin{bmatrix}
    		\dot{q} \\ \dot{p}
    	\end{bmatrix}
    	=
    	\begin{bmatrix}
    		0_{2\times 2} & I_2 \\
    		- I_2 & 0_{2\times 2}
    	\end{bmatrix}
    	\begin{bmatrix}
    		\nabla_q H \\ \nabla_p H
    	\end{bmatrix}
    	+
    	\begin{bmatrix}
    		0_{2\times 1} \\ \begin{matrix}
    			1 \\ 0
    		\end{matrix}
    	\end{bmatrix}
    	(u-d),
    \end{equation}
    where $H$ was defined in~(\ref{dyn}), and the disturbance $d\in\mathbb{R}$ is assumed to be of the form
    \begin{align*}
    	d = f^\top(q,p)\theta,
    \end{align*}
    with $f\in\mathbb{R}^\ell$ known and $\theta\in\mathbb{R}^\ell$ an unknown constant parameter vector.
    The robust energy shaping controller is given in the following lemma.
    
    \begin{lem}\label{l1}
    	Consider the rotary inverted pendulum with dynamics given in~(\ref{dis}). 
    	By applying
    	\begin{subequations}
    		\begin{align}
    			u &= u_{th1}+f^\top \hat{\theta},\label{u}\\
    			\dot{\hat{\theta}}&=-\tilde{p}_1\Gamma^{-1}f,\label{hat}
    		\end{align}
    	\end{subequations}
    	where $u_{th1}$ denotes the controller proposed in Theorem~\ref{th1} and $\Gamma\in\mathbb{R}^{\ell\times \ell}$ is a positive definite gain matrix, the equilibrium point $q^\ast$ is stabilized.
    \end{lem}
    \hfill$\square$
    
    \proof
    Applying the control law~(\ref{u}), where $u_{th1}$ is given by~(\ref{4}) with parameters defined in~(\ref{j2})--(\ref{psi}), to~(\ref{dis}), the closed-loop system can be written as
    \begin{align*}
    	\renewcommand{\arraystretch}{1.3}
    	\begin{bmatrix}
    		\dot{q} \\ \dot{p}
    	\end{bmatrix}
    	=
    	\begin{bmatrix}
    		0_{2\times 2} & M^{-1}M_d \\
    		- M_d M^{-1} & J_2 - G K_v G^\top
    	\end{bmatrix}
    	\begin{bmatrix}
    		\nabla_q H_d \\ \nabla_p H_d
    	\end{bmatrix}
    	+
    	\begin{bmatrix}
    		0_{2\times 1} \\ \begin{matrix}
    			f^\top\tilde{\theta} \\ 0
    		\end{matrix} 
    	\end{bmatrix},
    \end{align*}
   
    where $\tilde{\theta}=\hat{\theta}-\theta$, and $H_d$ and $J_2$ are defined in Theorem~\ref{th1}.  
     Consider the Lyapunov candidate function
    \begin{align*}
    	\mathcal{V}=H_d+\frac{1}{2}\tilde{\theta}^\top \Gamma^{-1}\tilde{\theta}.
    \end{align*}
    Using~(\ref{vdot}), its time derivative along the closed-loop trajectories is given by
    \begin{align*}
    	\dot{\mathcal{V}}
    	&=-(\nabla_p H_d)^\top G K_v G^\top \nabla_p H_d
    	+(\nabla_p H_d)^\top [f^\top\tilde{\theta},\,0]^\top
    \\
    &	+\tilde{\theta}^\top\Gamma^{-1}\dot{\tilde{\theta}}=-k_v\tilde{p}_1^2+\tilde{p}_1\tilde{\theta}^\top f
    	+\tilde{\theta}^\top\Gamma^{-1}\dot{\hat{\theta}},
    \end{align*}
    where $\tilde{p}=M_d^{-1}p=\nabla_p H_d$ and the constancy of $\theta$ has been used. 
    Substituting~(\ref{hat}) into the above expression yields
    \begin{align*}
    	\dot{\mathcal{V}}=-k_v\tilde{p}_1^2,
    \end{align*}
    which implies stability of $q^\ast$.
    \hfill$\blacksquare$

    \begin{remark}\label{re3}
    	In Lemma~\ref{l1}, a general class of matched disturbances is rejected by augmenting the control law with a robust term. 
    	In particular, the proposed approach includes the rejection of constant external disturbances, which commonly arise in practical applications. 
    	In this case, one has $f=1$, and the term $f^\top \hat{\theta}$ in~(\ref{u}) effectively acts as an integrator driven by the passive output of the closed-loop system. 
    	Moreover, the disturbance $d$ may also represent uncertainties in the RIP dynamics arising from unknown parameters or unmodeled dynamics.
    	On the other hand, Lemma~\ref{l1} generalizes the results reported in~\cite{franco2025integral}, where it is assumed that the regressor vector $f$ is integrable, that is, there exists a function $\zeta$ such that $\nabla \zeta = f$. 
    	In contrast, no such integrability assumption is required here. 
    	As a result, the controller proposed in Lemma~\ref{l1} is more broadly applicable and enhances robustness in practical implementations.
    \end{remark}
   \begin{figure}[t]
  	\centering
  	\includegraphics[scale=.54]{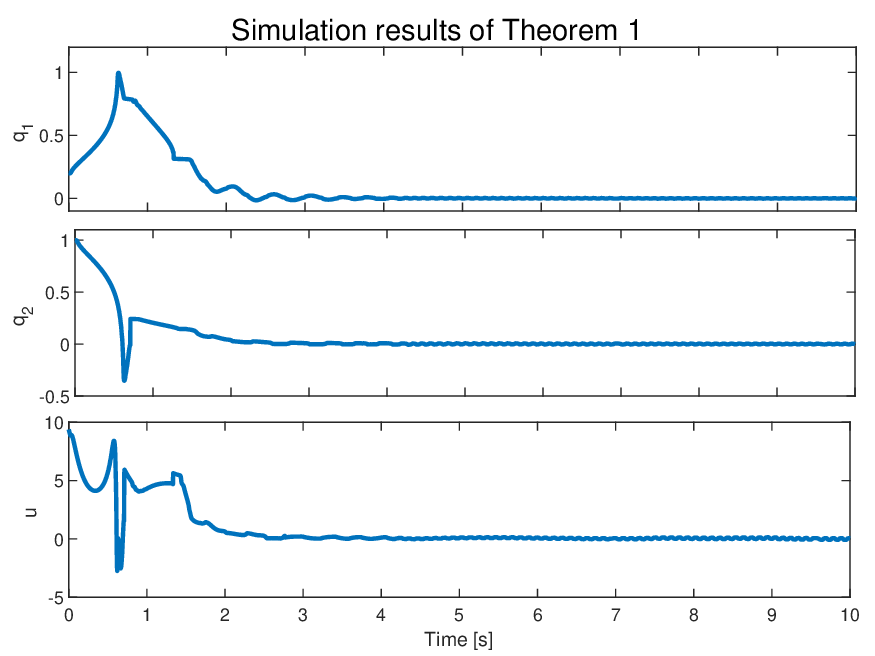}
  	\caption
  	{Simulation results of Theorem~\ref{th1}. The configuration variables converge to $q^\ast$ that means the upright position is stabilized.}
  	\label{sh2}
  \end{figure}   
    \section{Simulation Results}
    To examine the performance of Theorem~\ref{th1} and Lemma~\ref{l1}, they are implemented on RIP via simulation. The parameters of the robot are selected similar to \cite{hernandez2024modeling}. The controller's gains are $\psi_{4_0}=1, k_1=0.1$ and $k_2=100$. By this means, The positive definiteness region of $M_d$ reported in (\ref{rho}) is $\rho\approx 1.1$. Hence, the initial condition is chosen as $q_0=[0.2,1]^\top$ with zero velocity. The results of applying Theorem~\ref{th1} to the system are shown in Fig.~\ref{sh2}. Both of the configuration variables converge to  $q^\ast=[0,0]^\top$ in appropriate time. This approves the assertion of Theorem~\ref{th1} that by applying an energy shaping controller with parameters given in (\ref{j2})-(\ref{alpha}), the upright equilibrium is stabilized.
 
 Furthermore, a matched disturbance is applied to the RIP to test the performance of the proposed robust controller. Now, $q_0=[-0.8,0.8]^\top$, and 
 \begin{align*}
 	d&=0.1+0.1q_1-.3\sin(q_2)\cos(p_1)=f^\top\theta\\&=[1,q_1,\sin(q_2)\cos(p_1)][0.1,0.1,-0.3]^\top.
 \end{align*}
 In this case, both controllers of Theorem~\ref{th1} and Lemma~\ref{l1} are simulated. Fig.~\ref{sh3} shows the results of applying Theorem~\ref{th1}. It is apparent that $q$ does not converge to $q^\ast$ and the error remains in an ultimate bound. This is consistent with \cite{harandi2021robust} that proves matched disturbance leads to ultimate bound in the response of IDA-PBC. The results of Lemma~\ref{l1} are depicted in Fig.~\ref{sh4}. By applying the robust controller, the disturbance is rejected and configuration variables converge to $q^\ast$ that confirm Lemma~\ref{l1}. Note that as illustrated in Fig.~\ref{sh4}, estimation of external disturbance converges to $d$. This is fundamental result of adaptive estimators that under asymptotic stability condition, $f\tilde{\theta}=\hat{d}-d$ converges to zero whereas $\tilde{\theta}$ may be nonzero~\cite{karagiannis2006output}.

     \begin{figure}[t]
    	\centering
    	\includegraphics[scale=.54]{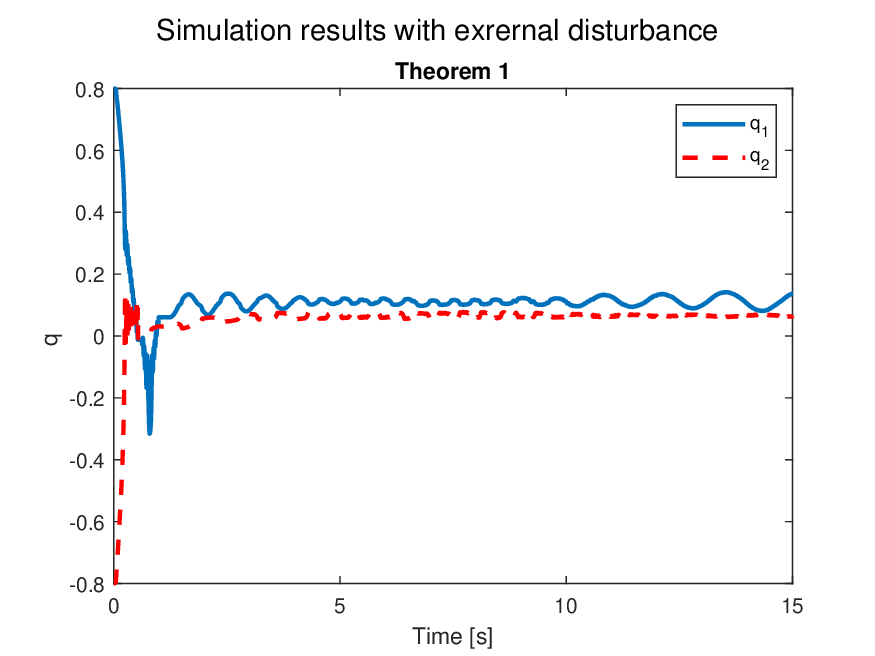}
    	\caption
    	{Simulation results of Theorem~\ref{th1} in the presence of external disturbance. A steady-state error remains in the configuration variables.}
    	\label{sh3}
    \end{figure}  
         \begin{figure}[h]
    	\centering
    	\includegraphics[scale=.54]{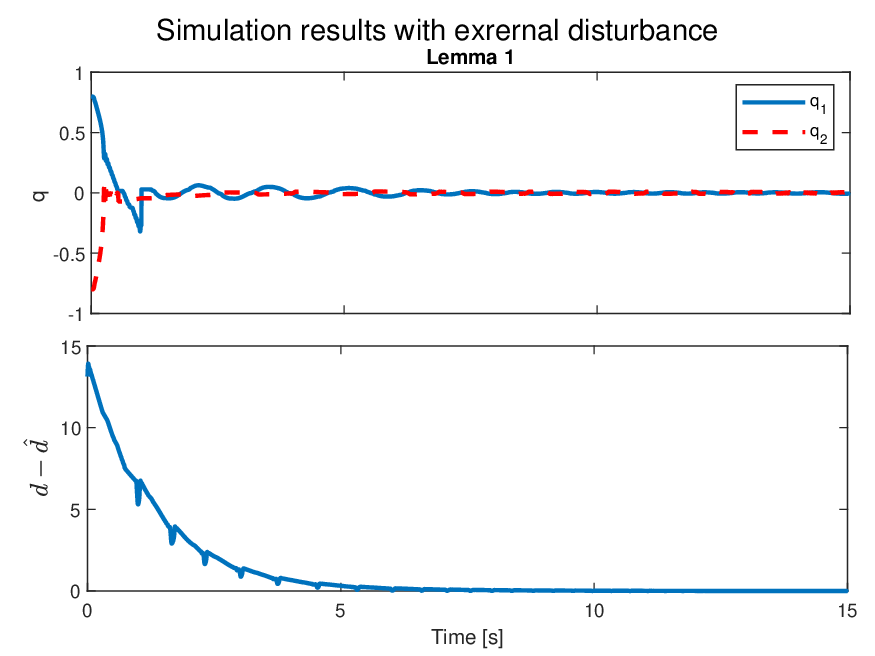}
    	\caption
    	{Simulation results of Lemma~\ref{l1} in the presence of external disturbance. Configuration variables converge to zero using the proposed robust controller.}
    	\label{sh4}
    \end{figure} 
    
    \section{Conclusion}
    This paper proposed a total energy shaping controller for rotary inverted pendulum to stabilize the upright position. For this purpose, a precise solution for the PDE of kinetic and potential energy was derived analytically. Furthermore,  previous works were analyzed and it was shown that their controllers do not satisfy the kinetic energy matching equation. Besides, a robust term was added to the controller to reject a particular type of external disturbance. The results were finally verified via simulations. The future works include implementation of the controller on RIP, and generalizing the proposed disturbance rejection approach for port Hamiltonian systems. 
\bibliographystyle{ieeetr}
\bibliography{reff}

@inproceedings{harandi2021robust,
	title={Robust IDA-PBC for a Spatial Underactuated Cable Driven Robot with Bounded Inputs},
	author={Harandi, M Reza J and Khalilpour, S Ahmad and Taghirad, Hamid},
	booktitle={2021 29th Iranian Conference on Electrical Engineering (ICEE)},
	pages={689--694},
	year={2021},
	organization={IEEE}
}

@article{viola2013some,
	title={Some Remarks on Interconnection and Damping Assignment Passivity-Based Control of Mechanical Systems},
	author={Viola, Giuseppe and Banavar and Acosta, JA and Astolfi, A},
	journal={Taming Heterogeneity and Complexity of Embedded Control},
	pages={721--735},
	year={2013},
	publisher={Wiley Online Library}
}

@article{hernandez2024modeling,
	title={Modeling, simulation, and control of a rotary inverted pendulum: A reinforcement learning-based control approach},
	author={Hernandez, Ruben and Garcia-Hernandez, Ramon and Jurado, Francisco},
	journal={Modelling},
	volume={5},
	number={4},
	pages={1824--1852},
	year={2024},
	publisher={MDPI}
}

@article{franco2025integral,
	title={Integral IDA-PBC for underactuated mechanical systems with unmeasured actuator dynamics and time-varying matched disturbances},
	author={Franco, Enrico and Chen, Kaiwen},
	journal={European Journal of Control},
	pages={101256},
	year={2025},
	publisher={Elsevier}
}

@inproceedings{ryalat2013ida,
	title={IDA-PBC for a class of underactuated mechanical systems with application to a rotary inverted pendulum},
	author={Ryalat, Mutaz and Laila, Dina Shona},
	booktitle={52nd IEEE conference on decision and control},
	pages={5240--5245},
	year={2013},
	organization={IEEE}
}

@article{donaire2017robust,
	title={Robust IDA-PBC for underactuated mechanical systems subject to matched disturbances},
	author={Donaire, Alejandro and Romero, Jose Guadalupe and Ortega, Romeo and Siciliano, Bruno and Crespo, Martin},
	journal={International Journal of Robust and Nonlinear Control},
	volume={27},
	number={6},
	pages={1000--1016},
	year={2017},
	publisher={Wiley Online Library}
}

@article{ferguson2020matched,
	title={On matched disturbance suppression for port-hamiltonian systems},
	author={Ferguson, Joel and Wu, Dongjun and Ortega, Romeo},
	journal={IEEE Control Systems Letters},
	volume={4},
	number={4},
	pages={892--897},
	year={2020},
	publisher={IEEE}
}

@inproceedings{borja2015shaping,
	title={Shaping the energy of port-Hamiltonian systems without solving PDE's},
	author={Borja, Pablo and Cisneros, Rafael and Ortega, Rom{\'e}o},
	booktitle={2015 54th IEEE Conference on Decision and Control (CDC)},
	pages={5713--5718},
	year={2015},
	organization={IEEE}
}

@incollection{karagiannis2006output,
	title={Output feedback stabilization of a class of uncertain systems},
	author={Karagiannis, D and Astolfi, A and Ortega, R},
	booktitle={Nonlinear and Adaptive Control: Tools and Algorithms for the User},
	pages={55--77},
	year={2006},
	publisher={World Scientific}
}

@phdthesis{harandi2021passivity,
	title={Passivity Based Control of 3-DOF
Underactuated Suspended Cable-Driven
Robot},
author={Harandi, M Reza J},
year={2021},
school={ K. N. Toosi University of Technology}
}

@article{harandi2022solution,
	title={Solution of matching equations of IDA-PBC by Pfaffian differential equations},
	author={Harandi, M Reza J and Taghirad, Hamid D},
	journal={International Journal of Control},
	volume={95},
	number={12},
	pages={3368--3378},
	year={2022},
	publisher={Taylor \& Francis}
}

@article{harandi2021matching,
	title={On the matching equations of kinetic energy shaping in ida-pbc},
	author={Harandi, M Reza J and Taghirad, Hamid D},
	journal={Journal of the Franklin Institute},
	year={2021},
	publisher={Elsevier}
}

@article{sandoval2008interconnection,
	title={Interconnection and damping assignment passivity—based control of the pendubot},
	author={Sandoval, Jes{\'u}s and Ortega, Romeo and Kelly, Rafael},
	journal={IFAC Proceedings Volumes},
	volume={41},
	number={2},
	pages={7700--7704},
	year={2008},
	publisher={Elsevier}
}

@article{viola2007total,
	title={Total energy shaping control of mechanical systems: simplifying the matching equations via coordinate changes},
	author={Viola, Giuseppe and Ortega, Romeo and Banavar, Ravi and Acosta, Jos{\'e} {\'A}ngel and Astolfi, Alessandro},
	journal={IEEE Transactions on Automatic Control},
	volume={52},
	number={6},
	pages={1093--1099},
	year={2007},
	publisher={IEEE}
}

@article{acosta2005interconnection,
	title={Interconnection and damping assignment passivity-based control of mechanical systems with underactuation degree one},
	author={Acosta, Jose Angel and Ortega, Romeo and Astolfi, Alessandro and Mahindrakar, Arun D},
	journal={IEEE Transactions on Automatic Control},
	volume={50},
	number={12},
	pages={1936--1955},
	year={2005},
	publisher={IEEE}
}

@inproceedings{acosta2009pdes,
	title={On the PDEs arising in IDA-PBC},
	author={Acosta, Jos{\'e} {\'A}ngel and Astol, A},
	booktitle={Proceedings of the 48h IEEE Conference on Decision and Control (CDC) held jointly with 2009 28th Chinese Control Conference},
	pages={2132--2137},
	year={2009},
	organization={IEEE}
}

@article{romero2018global,
	title={Global stabilisation of underactuated mechanical systems via PID passivity-based control},
	author={Romero, Jose Guadalupe and Donaire, Alejandro and Ortega, Romeo and Borja, Pablo},
	journal={Automatica},
	volume={96},
	pages={178--185},
	year={2018},
	publisher={Elsevier}
}

@inproceedings{donaire2015shaping,
	title={Shaping the energy of mechanical systems without solving partial differential equations},
	author={Donaire, Alejandro and Mehra, Rachit and Ortega, Romeo and Satpute, Sumeet and Romero, Jose Guadalupe and Kazi, Faruk and Singh, Navdeep M},
	booktitle={2015 American Control Conference (ACC)},
	pages={1351--1356},
	year={2015},
	organization={IEEE}
}

@article{ortega2002stabilization,
	title={Stabilization of a class of underactuated mechanical systems via interconnection and damping assignment},
	author={Ortega, Romeo and Spong, Mark W and G{\'o}mez-Estern, Fabio and Blankenstein, Guido},
	journal={IEEE transactions on automatic control},
	volume={47},
	number={8},
	pages={1218--1233},
	year={2002},
	publisher={IEEE}
}

\end{document}